\documentclass{article}
\usepackage{ijcai24}

\usepackage{times}
\usepackage{soul}
\usepackage{url}
\usepackage[hidelinks]{hyperref}
\usepackage[utf8]{inputenc}
\usepackage[small]{caption}
\usepackage{graphicx}
\usepackage{amsmath}
\usepackage{amsthm}
\usepackage{booktabs}
\usepackage[switch]{lineno}

\makeatletter
\newcommand{\dontusepackage}[2][]{%
	\@namedef{ver@#2.sty}{9999/12/31}%
	\@namedef{opt@#2.sty}{#1}}
\makeatother

\dontusepackage{algorithmic}

\usepackage{microtype}
\usepackage{graphicx}
\usepackage{booktabs} %
\usepackage{multirow}
\usepackage[shortlabels]{enumitem}
\usepackage{aliascnt}
\usepackage{makecell}
\usepackage[table]{xcolor}
\usepackage{pifont}%
\newcommand{\cmark}{\ding{51}}%
\newcommand{\xmark}{\ding{55}}%

\usepackage[ruled,noend,linesnumbered]{algorithm2e}

\SetKwProg{Function}{function}{}{}
\DontPrintSemicolon

\SetKwFor{ForEach}{for each}{do}{endfor}
\SetKwBlock{Loop}{loop}{end}
\SetFuncSty{textsc}

\usepackage{amsmath}
\usepackage{amssymb}
\usepackage{mathtools}
\usepackage{amsthm}
\usepackage[tight-spacing=true]{siunitx}

\usepackage[capitalize,noabbrev]{cleveref}

\theoremstyle{plain}
\newtheorem{theorem}{Theorem}[section]

\theoremstyle{definition}

\theoremstyle{remark}

\makeatletter
\let\ftype@table\ftype@figure
\let\ftype@algorithm\ftype@figure
\makeatother

\usepackage{amssymb}
\usepackage{amsmath}
\usepackage{mathtools}
\usepackage{physics}
\usepackage{xspace}
\usepackage{bm}
\usepackage{enumitem}
\usepackage{tikz}
\usetikzlibrary{positioning}
\usepackage{forest}
\usepackage{tikz-qtree}
\usepackage{linegoal}
\usepackage{natbib}
\let\cite\citep
\usepackage{autonum} %
\makeatletter
\autonum@generatePatchedReferenceCSL{Cref}
\autonum@generatePatchedReferenceCSL{cref}
\makeatother

\newcommand{\delimit}[3]{\newcommand{#1}[1]{\left#2##1\right#3}}
\delimit \ceil \lceil \rceil
\delimit \floor \lfloor \rfloor

\DeclareMathOperator*{\argmax}{argmax}
\DeclareMathOperator*{\E}{\mathbb E}

\let\op\operatorname
\let\eps\varepsilon

\let\tilde\widetilde

\newcommand{\R}{\mathbb R}

\newcommand{\ie}{{\em i.e.}\xspace}
\newcommand{\eg}{{\em e.g.}\xspace}
\newcommand{\etc}{{\em etc.}\xspace}

\newcommand{\cell}[2][c]{\begin{tabular}{#1}#2\end{tabular}}

\newcommand{\poly}{\op{poly}}

\definecolor{p1color}{RGB}{31,119,180}
\definecolor{p2color}{RGB}{255,127,14}
\definecolor{p3color}{RGB}{44,160,44}
\definecolor{p4color}{RGB}{214,39,40}

\tikzset{
  every node/.style={circle, draw, inner sep=0pt, minimum size=18pt},
  every path/.style={-to},
  label/.style={rectangle, draw=none, fill=white, inner sep=1pt, minimum size=0pt},
  terminal/.style={rectangle},
  node distance=1.6cm and 2cm,
}

\begin{document}
\title{Exponential Lower Bounds on the Double Oracle Algorithm in Zero-Sum Games}
\author{
Brian Hu Zhang$^1$
\and
Tuomas Sandholm$^{1,2,3,4}$
\affiliations
$^1$Computer Science Department, Carnegie Mellon University\\
$^2$Strategy Robot, Inc.\\
$^3$Strategic Machine, Inc. \\
$^4$Optimized Markets, Inc. \\
\emails
\{bhzhang, sandholm\}@cs.cmu.edu,
}
\maketitle
\begin{abstract}
The double oracle algorithm is a popular method of solving games, because it is able to reduce computing equilibria to computing a series of best responses. However, its theoretical properties are not well understood. In this paper, we provide exponential lower bounds on the performance of the double oracle algorithm in both partially-observable stochastic games (POSGs) and extensive-form games (EFGs). Our results depend on what is assumed about the {\em tiebreaking scheme}---that is, which meta-Nash equilibrium or best response is chosen, in the event that there are multiple to pick from. In particular, for EFGs, our lower bounds require {\em adversarial} tiebreaking, whereas for POSGs, our lower bounds apply regardless of how ties are broken. 
\end{abstract}
\section{Introduction}
The {\em double oracle algorithm}~\cite{McMahan03:Planning} is a popular practical framework for solving large games. It works by maintaining a {\em meta-game} comprised of a set of policies for each player, computing a {\em meta-Nash equilibrium} of the meta-game, and then computing {\em best responses} to that meta-game and adding those best responses to the meta-game for the next iteration. In essence, it reduces solving {\em multi-player} games to solving a series of small meta-games, which are easy, and {\em best-response problems}, which are single-player games. The method (or, more specifically, variations on the deep generalization of it, in which the best responses are replaced with deep RL-based approximate POMDP solvers, commonly referred to as a special case of the {\em policy-space response oracle}~\cite{Lanctot17:Unified} algorithm), has been successfully applied to large, two-player zero-sum games such as {\em Barrage Stratego}~\cite{McAleer20:Pipeline} and {\em StarCraft}~\cite{Vinyals19:Grandmaster}.  In practice, the algorithm tends to converge very fast: even in games far too large to enumerate  the state space, only tens or hundreds of iterations are required to reach strong play. 

However, to our knowledge, the theoretical properties of double oracle are almost completely unstudied. Indeed, the lack of an efficient convergence guarantee has led to several variants of double oracle being developed which {\em do} have efficient convergence guarantees, most notably the {\em sequence-form} \cite{Bosansky14:Exact} and {\em extensive-form double oracle}~\cite{McAleer21:Xdo} algorithms. In extensive-form games, both of these algorithms are guaranteed to converge in a number of iterations polynomial in the size of the game. Another variant of double oracle, {\em self-play PSRO}~\cite{McAleer22:Self} has also been developed that adds {\em randomized} policies to the meta-game, in the hopes that such policies lead to faster learning. In this paper, however, we focus on the plain version of the double oracle algorithm.\footnote{In multi-player general-sum games, especially when the game is large enough that ``best'' responses are approximated with deep reinforcement learning, generalizations and variants of the double oracle algorithm have been studied under the name {\em policy space response oracle} (PSRO)~[\eg, \citealp{Lanctot17:Unified}]. In this paper, we adhere to the more traditional name {\em double oracle} because we are indeed working with the more ``standard'' two-player version of the algorithm, not any generalization thereof.}

We derive several different partially-observable stochastic games (POSGs) in which double oracle takes exponentially many iterations to converge. The games differ in their structure and in what assumptions need to be made about the choices left unspecified in the algorithm, namely, the choices of {\em initialization, meta-Nash equilibria}, and {\em best responses}. For example, if all choices are {\em random} then we give a {\em partially-observable stochastic game} with an exponential convergence bound (\Cref{th:posg}); if all choices can be made {\em adversarially}, then we give a {\em tree-form, fully-observable} game (\Cref{th:efg}). A summary of our results can be found in \Cref{tab:summary}.

\section{Preliminaries}
A {\em two-player partially-observable stochastic game} (POSG) (hereafter simply {\em game}) consists of the following elements:\footnote{The definition used here is more restrictive than many common definitions of POSGs. For example, many authors allow observations to be randomized, or action sets to depend on state, or rewards to be given at nonterminal states and be action-dependent. But since this whole paper concerns only {\em lower bounds}, adding restrictions makes our results {\em more} powerful. It also simplifies our notation.}
\begin{enumerate}
\item A finite {\em state space} $S$, {\em action spaces} $A_1, A_2$, and {\em observation space} $O$ with $|O| \le |S|$;
\item a {\em starting distribution} $S_0 \in \Delta(S)$;
\item a set of {\em terminal states} $Z \subset S$;
\item for each $(s, a_1, a_2)$ where $s \in S \setminus Z , a_1 \in A_1, a_2 \in A_2$, a probability distribution $p(\cdot | s, a_1, a_2) \in \Delta(S)$ denoting the probability of transitioning to the next state; 
\item two {\em observation function} $o_1, o_2 : S \setminus Z \to O$; and
\item two {\em reward functions} $R_1, R_2 : Z \to [-1, +1]$ denoting the reward of P1 and P2 respectively, as a function of the terminal state reached.
\end{enumerate}
A game is {\em zero-sum} if $R_1 = -R_2$.
We will make the assumption that the game has a DAG structure: the transition multigraph of the game---that is, the multigraph whose nodes are the states and for which there is an edge $(s, s')$ for each pair $(a_1, a_2) \in A_1 \times A_2$ such that $p(s'|s, a_1, a_2) > 0$---is directed and acyclic. Thus, the terminal states $z \in Z$ are the sinks of this DAG. We will denote the depth of the DAG by $k$. 

A {\em pure policy} for a player $i \in \{1, 2\}$ is a mapping $\pi_i : O^{\le d} \to A_i$, where $O^{\le d}$ denotes the set of sequences on $O$ of length at most $k$. We denote by $\Pi_i$ the set of pure policies of player $i$. A pair of pure policies $(\pi_1, \pi_2)$ is a {\em policy profile} or simply {\em profile}. A profile induces a distribution over the terminal states $Z$ of the game, given by sampling $s_0 \sim S_0$ and then following $(\pi_1, \pi_2)$ until a state $z \in Z$ is reached. We will use $z \sim (\pi_1, \pi_2)$ to denote a sample from this distribution. A {\em mixed policy} $\mu_i \in \Delta(\Pi_i)$ is a distribution over pure policies. Given mixed profile $(\mu_1, \mu_2)$, the {\em expected value} of player $i$ is $$V_i(\mu_1, \mu_2) = \E_{\substack{\pi_1 \sim \mu_1,\\ \pi_2 \sim \mu_2,\\ z \sim (\pi_1, \pi_2)}} R_i(z).$$
Policy $\pi_i \in \Pi_i$ is a {\em best response} to a mixed policy $\mu_{-i}$ if $$\pi_i \in \argmax_{\pi_i' \in \Pi_i} V_i(\pi_i', \mu_{-i}).$$
An {\em $\eps$-Nash equilibrium} is a profile $(\mu_1, \mu_2)$ such that neither player can improve by more than $\eps$:
\begin{align}
\max_{\pi_i \in \Pi_i} V_i(\pi_i, \mu_{-i}) - V_i(\mu) \le \eps.
\end{align}
A Nash equilibrium is a $0$-Nash equilibrium. In general, computing a Nash equilibrium of a POSG is hard---indeed, even solving POMDPs (\ie, POSGs where $|A_2| = 1$) is PSPACE-complete~\cite{Papadimitriou87:Complexity}. It will be useful to define several special cases of POSGs:
\begin{enumerate}
\item A {\em (fully-observable) stochastic game} is a POSG in which both players observe the true state, \ie, $S = O$ and $o_1(s) = o_2(s) = s$.
\item A {\em tree-form game} is a POSG in which the transition multigraph is a tree. 
\item A {\em normal-form game} is a stochastic game with a single nonterminal state (which is also the start state). A two-player normal-form game is described by two matrices $V, V_2 \in \R^{A_1 \times A_2}$, where $V_i(a_1, a_2)$ is the reward to player $i$ if P1 plays action $a_1$ and P2 plays $a_2$.
\end{enumerate}
Any of the other forms can be converted into normal form at the cost of a larger game: namely, any POSG is equivalent to the normal-form game described by matrices $V_1, V_2 \in \R^{\Pi_1 \times \Pi_2}$. This conversion, however, incurs doubly-exponential blowup in the size of the game in general. 

For zero-sum games, in each of the special cases, there are polynomial-time algorithms for exactly computing a Nash equilibrium: in the fully-observable case, one can perform backwards induction (value iteration) starting from the leaves, solving each state via a linear program; tree-form POSGs are a subclass of {\em extensive-form games}, and \citet{Koller94:Fast} describe an LP-based method that runs in polynomial time.
\subsection{The Double Oracle Algorithm}
\begin{algorithm}
\caption{The double oracle algorithm. {\sc NormalFormNashEquilibrium} returns an exact Nash equilibrium to the normal-form game in which each player picks a policy from its policy set $\tilde \Pi_i$. {\sc BestResponse} returns a pure policy that is a best response to the given opponent policy.}\label{alg:do}
{\bf Input:} POSG, initial strategies $\pi_1^0 \in \Pi_1, \pi_2^0 \in \Pi_2$, \\\qq{} desired Nash gap $\eps \ge 0$\\
{\bf Output:} $\eps$-Nash equilibrium $(\mu_1, \mu_2)$ of the POSG\\
$\tilde\Pi_1^0 \gets \{ \pi_1^0 \}, \tilde\Pi_2^0 \gets \{ \pi_2^0 \}$\\
\For{$t = 1, 2, \dots$}{
$\mu_1^t, \mu_2^t \gets$ Nash equilibrium of \\ \qq{}normal-form game $(\tilde \Pi_1^{t-1}, \tilde \Pi_2^{t-1})$\\
$\pi_1^{t} \gets$ P1 best response to $\mu_2^t$\\
$\pi_2^{t} \gets$ P2 best response to $\mu_1^t$\\
\lIf{Nash gap $\le \eps$}{\Return{$(\mu_1^t, \mu_2^t)$}}
$\tilde \Pi_1^{t} \gets \tilde \Pi_1^{t-1} \cup \{ \pi_1^{t}\}$\\
$\tilde \Pi_2^{t} \gets \tilde \Pi_1^{t-1} \cup \{ \pi_2^{t}\}$
}
\end{algorithm} 
Pseudocode for the double oracle algorithm is given in \Cref{alg:do}. The algorithm is simple: it iteratively maintains a {\em meta-game} $(\tilde \Pi_1, \tilde \Pi_2)$, computes a {\em meta-Nash equilibrium} $(\mu_1, \mu_2)$ to that meta-game, computes best responses $(\pi_1, \pi_2)$ in the full game, and adds those best responses to the meta-game. Double oracle clearly converges in a finite number of steps: there are only a finite number of pure policies, and each iteration of the main loop must add a pure policy to at least one player's meta-game policy set (if both best responses $\pi_1, \pi_2$ are already in the policy sets, then the Nash gap would be $0$). 

The meta-game on iteration $t$ is a $t \times t$ normal-form game. For zero-sum games at least, as specified above, Nash equilibria can be easily computed in polynomial time via linear programming~\cite{vonNeumann28:Zur}. Thus, the entire complexity of \Cref{alg:do} lies in the best responses (which are POMDPs) and the number of iterations $t$ until the algorithm terminates. For nonzero-sum games, Nash equilibrium computation is in general hard~\cite{Chen09:Settling}. However, we will ignore these computational issues and focus our attention on the number of iterations it takes for double oracle to converge.

The double oracle algorithm is not affected by the game representation. For example, running double oracle on a POSG and running double oracle on the normal form of that POSG would produce the same result. Therefore, for the rest of the paper, we will call two games {\em (strategically) equivalent} if they induce the same normal form.

\section{Main Results}

\begin{table*}[t]
\centering
\scalebox{1}{
\begin{tabular}{cccccccccc}
& \multicolumn{3}{c}{game properties} & \multicolumn{3}{c}{double oracle assumptions}
\\
& ZS & FO & TF & Nash support & initialization & meta-Nash & best responses & $|S|$ & $\eps^*$ \\\midrule
\makecell{\Cref{th:gmp}} & \cmark & \cmark & \xmark & $2^{\Theta(k)}$ & --- & --- & --- & $O(k)$ & $2^{-\Theta(k)}$ \\
\makecell{\Cref{th:posg}} & \cmark & \xmark & \xmark & $1$ & random & --- & --- & $O(k)$ & $\Theta(1)$ \\
\makecell{\Cref{th:sg}} & \cmark & \cmark & \xmark & $1$ & random & --- & adversarial & $O(k)$ & $\Theta(1)$ \\
\Cref{th:efg-nz} & \xmark & \xmark & \cmark & $1$ & adversarial & adversarial & --- & $\poly(k)$ & $\Theta(1/k)$\\
\Cref{th:efg} & \cmark & \cmark & \cmark & $2$ & adversarial & adversarial & adversarial & $O(k)$ & $\Theta(1/k)$
\end{tabular}
}
\caption{Summary of main results. {\em Nash support} gives the minimum support per player, in pure policies, of any exact Nash equilibrium.  In all cases double oracle takes $2^{\Theta(k)}$ iterations to converge to an $\eps$-equilibrium for every $\eps < \eps^*$. `ZS', `FO', and `TF' mean zero-sum, fully-observable, and tree-form, respectively.}\label{tab:summary}
\end{table*}

As suggested above, the main results in this paper are {\em lower bounds} on the complexity of the double oracle algorithm. In particular, we will give several game examples in which double oracle, under various assumptions about the best response oracle, fails to converge to an $\eps$-equilibrium, for moderately-sized $\eps$, until $t$ is exponentially large.

\begin{figure*}[p]
\centering
\begin{tikzpicture}[node distance = 1cm and 2cm]
\node(s0){};
\node(start)[draw=none, left=1em of s0]{};
\draw (start) -- (s0);
\node(s1)[right=of s0]{};
\draw[bend left=30] (s0) to node[label,midway]{1,1} (s1);
\draw[bend right=30] (s0) to node[label,midway]{0,0} (s1);
\node(a1)[above=of s0,terminal]{+1};
\draw[bend left=30] (s0) to node[label,midway]{1,0} (a1);
\draw[bend right=30] (s0) to node[label,midway]{0,1} (a1);
\node(s2)[right=of s1]{};
\draw[bend left=30] (s1) to node[label,midway]{1,1} (s2);
\draw[bend right=30] (s1) to node[label,midway]{0,0} (s2);
\node(a2)[above=of s1,terminal]{+1};
\draw[bend left=30] (s1) to node[label,midway]{1,0} (a2);
\draw[bend right=30] (s1) to node[label,midway]{0,1} (a2);
\node(s3)[right=of s2]{};
\draw[bend left=30] (s2) to node[label,midway]{1,1} (s3);
\draw[bend right=30] (s2) to node[label,midway]{0,0} (s3);
\node(a3)[above=of s2,terminal]{+1};
\draw[bend left=30] (s2) to node[label,midway]{1,0} (a3);
\draw[bend right=30] (s2) to node[label,midway]{0,1} (a3);
\node(s4)[right=of s3,terminal]{--1};
\draw[bend left=30] (s3) to node[label,midway]{1,1} (s4);
\draw[bend right=30] (s3) to node[label,midway]{0,0} (s4);
\node(a4)[above=of s3,terminal]{+1};
\draw[bend left=30] (s3) to node[label,midway]{1,0} (a4);
\draw[bend right=30] (s3) to node[label,midway]{0,1} (a4);
\end{tikzpicture}
\caption{The $k$-bit guess-the-string game, here depicted for $k = 4$. The action spaces are $A_1 = A_2 = \{0, 1\}$. The start state is the leftmost state, labeled with $\to$. Terminal states are drawn as rectangles, and their rewards are written within them. Transitions are deterministic, and edges are labeled with the transitions that take them there. }\label{fig:support}
\end{figure*}
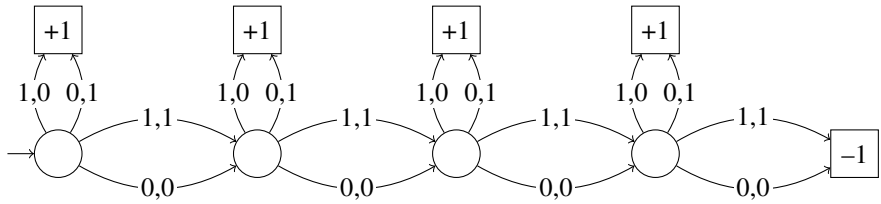

In general, a stochastic game may have no Nash equilibria with small support. For example, consider the $k$-bit ``generalized matching pennies'' game in which P1 picks a string $\pi_1 \in \{0, 1\}^k$ one bit at a time, and P2 simultaneously attempts to guess that string, also one bit at a time, with P2 winning if and only if P1 and P2 guess the same string. 

This game for $k = 4$ is depicted in \Cref{fig:support}. Intuitively, it is nothing more than a finite automaton that reads two bitstrings $a_1, a_2 \in \{0, 1\}^k$ (interpreted as natural numbers in $\{0, 1, \dots, 2^k-1\}$) in parallel, and outputs the reward $u(a_1, a_2)$ as specified by the normal-form game: that is, it returns $-1$ if the strings are equal and $+1$ otherwise. This proves:
\begin{theorem}\label{th:gmp}
    For every $k \ge 1$, there exists a zero-sum fully-observable stochastic game with $O(k)$ nodes in which, regardless of initialization, meta-Nash, or best responses, double oracle takes $2^{\Theta(k)}$ iterations to find an exact equilibrium.
\end{theorem}
However, the ``generalized matching pennies'' game is not ideal as a counterexample, for multiple reasons:
\begin{enumerate}
\item {\em Polynomial-time approximation}: While double oracle fails to converge to {\em exact} equilibrium in polynomially many iterations, it will converge to a $\eps$-equilibrium in $O(1/\eps)$ iterations: one can check inductively that, at odd iterations, P1 will add an arbitrary new policy to its support $\tilde \Pi_1$, and P2 will add the same policy at the next (even) iteration. Thus, after $2t$ iterations we will have $\tilde \Pi_1^{2t} = \tilde \Pi_2^{2t}$ and $(\mu^{2t}_1, \mu^{2t}_2)$ will be a $1/t$-equilibrium. This is still a reasonable convergence rate.
\item {\em High support}. As mentioned above, the game has only high-support equilibria.
\end{enumerate} 

The main counterexamples in our paper will fix both of these issues. In particular, all our counterexamples will be families of games in which {\em there is a Nash equilibrium with constant support size}, and yet double oracle {\em fails to find any $\eps$-approximate equilibrium} in $\poly(N, 1/\eps)$ iterations, where $N$ is the size of the representation of the POSG. These counterexamples are summarized in \Cref{tab:summary}.

\begin{theorem}\label{th:posg}
For every $k \ge 1$, there exists a zero-sum POSG with $O(k)$ states and a pure Nash equilibrium in which, in the double oracle algorithm,
\begin{itemize}
\item the meta-Nash equilibria and the best responses are unique on every iteration, and
\item for $\eps$ constant, if the starting policies $\pi_1^0, \pi_2^0$ are chosen uniformly at random\footnote{Choosing starting policies at random means choosing a {\em pure} policy $\pi_1^0$ from $\Pi_1$ uniformly at random, not setting $\pi_1^0$ to be the uniformly random policy.}, then double oracle takes $\Theta(2^k)$ iterations in expectation. 
\end{itemize}
\end{theorem}

\begin{proof}
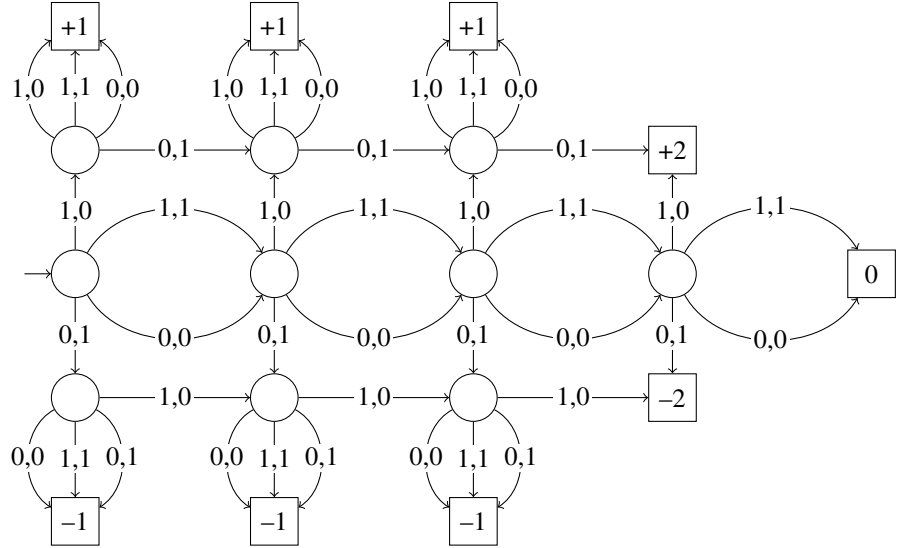
\begin{figure*}[p]
\centering
\begin{tikzpicture}[node distance = 1cm and 2cm]
\node(s0){};
\node(start)[draw=none, left=1em of s0]{};
\draw (start) -- (s0);
\node(s1)[right=of s0]{};
\draw[bend left=60] (s0) to node[label,midway]{1,1} (s1);
\draw[bend right=60] (s0) to node[label,midway]{0,0} (s1);
\node(a1)[above=of s0]{};
\node(b1)[below=of s0]{};
\draw (s0) -- (a1) node[label,midway]{1,0};
\draw (s0) -- (b1) node[label,midway]{0,1};
\node(s2)[right=of s1]{};
\draw[bend left=60] (s1) to node[label,midway]{1,1} (s2);
\draw[bend right=60] (s1) to node[label,midway]{0,0} (s2);
\node(a2)[above=of s1]{};
\node(b2)[below=of s1]{};
\draw (s1) -- (a2) node[label,midway]{1,0};
\draw (s1) -- (b2) node[label,midway]{0,1};
\draw (a1) -- (a2) node[label,midway]{0,1}; 
\draw (b1) -- (b2) node[label,midway]{1,0}; 
\node(aa2)[above=of a1,terminal]{+1};
\node(bb2)[below=of b1,terminal]{--1};
\draw (a1) to node[label,midway]{1,1} (aa2);
\draw[bend left=60] (a1) to node[label,midway]{1,0} (aa2);
\draw[bend right=60] (a1) to node[label,midway]{0,0} (aa2);
\draw (b1) to node[label,midway]{1,1} (bb2);
\draw[bend left=60] (b1) to node[label,midway]{0,1} (bb2);
\draw[bend right=60] (b1) to node[label,midway]{0,0} (bb2);
\node(s3)[right=of s2]{};
\draw[bend left=60] (s2) to node[label,midway]{1,1} (s3);
\draw[bend right=60] (s2) to node[label,midway]{0,0} (s3);
\node(a3)[above=of s2]{};
\node(b3)[below=of s2]{};
\draw (s2) -- (a3) node[label,midway]{1,0};
\draw (s2) -- (b3) node[label,midway]{0,1};
\draw (a2) -- (a3) node[label,midway]{0,1}; 
\draw (b2) -- (b3) node[label,midway]{1,0}; 
\node(aa3)[above=of a2,terminal]{+1};
\node(bb3)[below=of b2,terminal]{--1};
\draw (a2) to node[label,midway]{1,1} (aa3);
\draw[bend left=60] (a2) to node[label,midway]{1,0} (aa3);
\draw[bend right=60] (a2) to node[label,midway]{0,0} (aa3);
\draw (b2) to node[label,midway]{1,1} (bb3);
\draw[bend left=60] (b2) to node[label,midway]{0,1} (bb3);
\draw[bend right=60] (b2) to node[label,midway]{0,0} (bb3);
\node(s4)[right=of s3,terminal]{0};
\draw[bend left=60] (s3) to node[label,midway]{1,1} (s4);
\draw[bend right=60] (s3) to node[label,midway]{0,0} (s4);
\node(a4)[above=of s3,terminal]{+2};
\node(b4)[below=of s3,terminal]{--2};
\draw (s3) -- (a4) node[label,midway]{1,0};
\draw (s3) -- (b4) node[label,midway]{0,1};
\draw (a3) -- (a4) node[label,midway]{0,1}; 
\draw (b3) -- (b4) node[label,midway]{1,0}; 
\node(aa4)[above=of a3,terminal]{+1};
\node(bb4)[below=of b3,terminal]{--1};
\draw (a3) to node[label,midway]{1,1} (aa4);
\draw[bend left=60] (a3) to node[label,midway]{1,0} (aa4);
\draw[bend right=60] (a3) to node[label,midway]{0,0} (aa4);
\draw (b3) to node[label,midway]{1,1} (bb4);
\draw[bend left=60] (b3) to node[label,midway]{0,1} (bb4);
\draw[bend right=60] (b3) to node[label,midway]{0,0} (bb4);
\end{tikzpicture}
\caption{The $2^k$-bigger-number game used in \Cref{th:posg}, here depicted for $k = 4$. Observations are trivial: $|O| = 1$.}\label{fig:posg}
\end{figure*}
\begin{figure*}[p]
\centering
\begin{tikzpicture}[node distance = 0.7cm and 2cm]
\node(s0){};
\node(start)[draw=none, left=1em of s0]{};
\draw (start) -- (s0);
\node(s1)[right=of s0]{};
\draw[bend left=60] (s0) to node[label,midway]{1,1} (s1);
\draw[bend right=60] (s0) to node[label,midway]{0,0} (s1);
\node(a1)[above=of s0,terminal]{+1};
\node(b1)[below=of s0,terminal]{--1};
\draw (s0) -- (a1) node[label,midway]{1,0};
\draw (s0) -- (b1) node[label,midway]{0,1};
\node(s2)[right=of s1]{};
\draw[bend left=60] (s1) to node[label,midway]{1,1} (s2);
\draw[bend right=60] (s1) to node[label,midway]{0,0} (s2);
\node(a2)[above=of s1,terminal]{+1};
\node(b2)[below=of s1,terminal]{--1};
\draw (s1) -- (a2) node[label,midway]{1,0};
\draw (s1) -- (b2) node[label,midway]{0,1};
\node(s3)[right=of s2]{};
\draw[bend left=30] (s2) to node[label,midway]{1,1} (s3);
\draw[bend right=30] (s2) to node[label,midway]{0,0} (s3);
\node(a3)[above=of s2,terminal]{+1};
\node(b3)[below=of s2,terminal]{--1};
\draw (s2) -- (a3) node[label,midway]{1,0};
\draw (s2) -- (b3) node[label,midway]{0,1};
\node(s4)[right=of s3,terminal]{0};
\draw[bend left=30] (s3) to node[label,midway]{1,1} (s4);
\draw[bend right=30] (s3) to node[label,midway]{0,0} (s4);
\node(a4)[above=of s3,terminal]{+1};
\node(b4)[below=of s3,terminal]{--1};
\draw (s3) -- (a4) node[label,midway]{1,0};
\draw (s3) -- (b4) node[label,midway]{0,1};
\end{tikzpicture}
\caption{The $2^k$-weak bigger-number game used in \Cref{th:sg}, here depicted for $k = 4$.}\label{fig:sg}
\end{figure*}
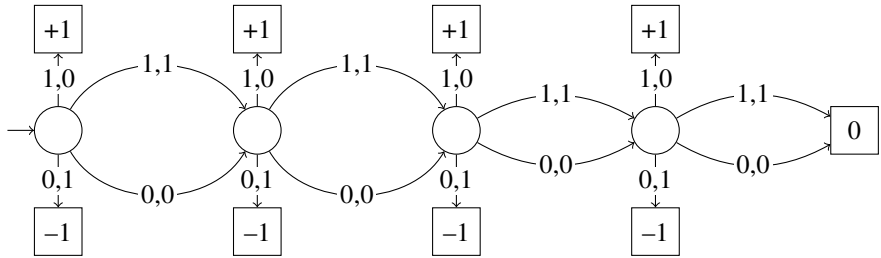

The proof is based on a simple normal-form game that we call the $n$-{\em bigger-number game}. In the $n$-bigger-number game, each player's action space is $A_1 = A_2 = [n] := \{0, \dots, n-1\}$, and the rules are as follows. Both players simultaneously select numbers $a_i \in [n]$. If $a_i = a_j$, then both players score $0$. Otherwise, the player who plays the bigger number scores $1$, unless $|a_i - a_j| = 1$ in which case they score $2$.

We first analyze the behavior of double oracle in the $n$-bigger-number game with random starting policies. For each $t$, let $M(t)$ be the largest number in the support of either player's policy set $\Pi_i^t$. Then with constant probability, $M(0) \le n/2$. Further,  $\mu_1^t$ is supported on $\{0, \dots, M(t-1)\}$. Then the best response $\pi_2^t$ to $\mu_1^t$ is at most $M(t-1) + 1$, because any number larger than $t$ performs worse than $\max \op{supp}(\mu_1^t) + 1 \le t$. Thus, $M(t) \le M(t-1) + 1$ for all $t \ge 1$. Equilibrium can only be reached when $M(t) = n$, because the only equilibrium of the game is $(n, n)$. Therefore, with constant probability, double oracle takes $\Theta(n)$ iterations, and therefore the expected number of iterations for double oracle is also $\Theta(n)$.

We now show that the $n$-bigger-number game, for $n = 2^k$, is equivalent to a POSG with $O(k)$ nodes, which would complete the proof. Consider the POSG depicted in \Cref{fig:posg} (for $k = 4$, easily generalizable). Like \Cref{fig:support}, this POSG is essentially a finite automaton that reads two bitstrings $a_1, a_2$ simultaneously, and outputs the required value. The reward depends on the value of $a_1 - a_2$, in particular, whether it is greater than 1, equal to 1, equal to 0, equal to -1, or less than -1. The center row of nodes captures the states in which the substrings read are currently equal (If that continues until the last timestep, then the numbers are equal). The row above the center captures the states in which $a_1 \ne a_2$ but it is still possible for $a_1 = a_2 + 1$. (This happens if $a_1 = x10^\ell$ and $a_2 = x01^\ell$ for some string $x$ and integer $\ell$.) The row below the center is the same but with the players flipped. 

Since observations are trivial, a pure policy in this POSG is specified by a vector $\pi_i \in \{0, 1\}^k$, whose $j$th index specifies the action played by player $i$ at time $j \in [k]$. The vector $\pi_i$ is then identified with the pure action in the $2^k$-bigger-number game whose binary representation is $\pi_i$. This POSG is equivalent to the $2^k$-bigger-number game. 
\end{proof}

The next two results will be similar to the above result, but will have increasingly stringent requirements on the structure of the game---first, stochastic games, and then tree-form stochastic games. In exchange, we will also need more stringent requirements on the behavior of the double oracle algorithm. In particular, the meta-Nash equilibria and best responses used by double oracle may no longer be unique, so we will need to make assumptions on how they are chosen. Whenever the choice is not unique, we will always assume {\em adversarial} choices for the algorithm---that is, we will assume that meta-Nash equilibria and best responses are chosen to make double oracle run for as long as possible.

\begin{theorem}\label{th:sg}
For every $k \ge 1$, there exists a zero-sum  {\em fully}-observable stocastic game with $O(k)$ states and a pure Nash equilibrium , in which, in the double oracle algorithm,
\begin{itemize}
\item the meta-Nash equilibria are unique on every iteration,
\item the best responses are {\em not} unique on every iteration, and
\item for $\eps < 2$, if the starting policies $\pi_1^0, \pi_2^0$ are chosen uniformly at random, double oracle {\em with adversarial best responses} takes $2^k-1$ iterations. 
\end{itemize}
\end{theorem}
\begin{proof}
We will define a $n$-{\em weak bigger-number game} similar to the $n$-bigger-number game used in the proof of \Cref{th:posg}. In the {\em weak $n$-bigger-number game}, two players simultaneously select a number $a_i \in [n]$, and whoever picks the bigger number wins (scores $1$). 

Unlike the bigger-number game, best responses will not be unique in the weak bigger-number game. For example, every number bigger than $0$ is a best response to $0$. However, we can still replicate the behavior of double oracle on the bigger-number game, because the same conditions for that behavior still hold: namely, the only Nash equilibrium is $(n, n)$, and $\max \op{supp}(\mu_1^t) + 1$ is always a best response to $\mu_1^t$. Therefore, if we always adversarially choose this best response, an identical analysis holds, and the expected runtime of double oracle is $\Theta(n)$ iterations.

We now only need to represent the $2^k$-weak bigger-number game as a stochastic game. Consider the stochastic game in \Cref{fig:sg}, which is this time a {\em fully-observable} game\footnote{The observations in this game are actually irrelevant, because there is only one possible state corresponding to each history length.}. Once again, this POSG is essentially a finite automaton that computes the game value. As before, we relate the policies, which are vectors $\pi_i \in \{0, 1\}^k$, to numbers in $\{0, \dots, 2^k-1\}$ via their binary representation, and from this it is easy to see that the normal form of this stochastic game is indeed the $2^k$-weak bigger-number game.
\end{proof}

Our next result is the only result that uses a {\em nonzero-sum game}, and the first of two results concerning {\em tree-form} games.
\begin{figure*}
\centering
\tikzset{
terminal/.style={rectangle, draw, inner sep=3pt},
}
\forestset{
  el/.style={edge label={node[midway, label] {#1}}},
  default preamble={for tree={ 
  l=1.6cm, child anchor=north}},
  nat/.style={draw=none,minimum size=0pt,edge={-}}
}
\scalebox{0.8}{
\begin{forest}
[
    [,el={\cell{0$^\ell$,0$^\ell$\\ 1$^\ell$,1$^\ell$}},edge={-},nat
        [,el={$i{<}k{-}\ell$},s sep=0.6cm
            [{$0$},el={\cell{0,0\\1,1}},terminal]
            [{$-1$},el={\cell{0,1\\1,0}},terminal]
        ]
    ]
    [,el={0$^\ell$,1$^\ell$},edge={-},nat
        [,el={$i{<}k{-}\ell$}
            [{$1/2k,-1$},el={\cell{0,0\\1,1}},terminal]
            [{$-1$},el={\cell{0,1\\1,0}},terminal]
        ]
    ]
    [,el={0$^\ell$,1},edge={-},nat
        [,el={$i{<}k{-}\ell$}
            [{$1/2k,-1$},el={\cell{0,0\\1,1}},terminal]
            [{$-1$},el={\cell{0,1\\1,0}},terminal]
        ]
        [,el={$i{=}k{-}\ell$}
            [{$-1,1/2k$},el={1,1},terminal]
            [{$-1$},el={1,0},terminal]
        ]
        [,el={$i{>}k{-}\ell$}
            [{$-1,1/2k$},el={0,0},terminal]
            [{$-1$},el={0,1},terminal]
        ]
    ]
    [,el={1$^\ell$,0},edge={-},nat
        [,el={$i{<}k{-}\ell$}
            [{$1/2k,-1$},el={\cell{0,0\\1,1}},terminal]
            [{$-1$},el={\cell{0,1\\1,0}},terminal]
        ]
        [,el={$i{=}k{-}\ell$}
            [{$1/2k,-1$},el={0,0},terminal]
            [{$-1$},el={0,1},terminal]
        ]
        [,el={$i{>}k{-}\ell$}
            [{$1/2k,-1$},el={1,1},terminal]
            [{$-1$},el={1,0},terminal]
        ]
    ]
]
\end{forest}
}
\caption{A depiction of the game used in \Cref{th:efg-nz}. Observations are not shown: the only nontrivial observation each player makes is the randomly-selected index $i$. Not all actions and transitions are shown. If a terminal node contains only one reward, then that is the reward of both players.}\label{fig:efg-nz}
\end{figure*}
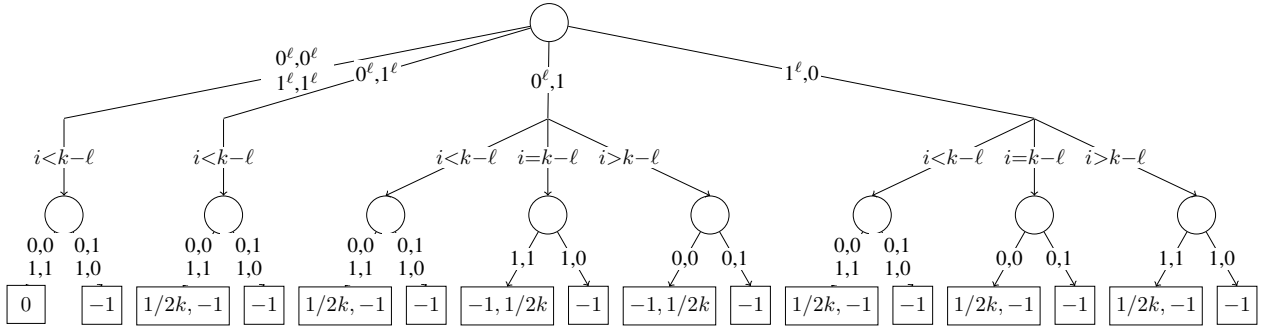
\begin{theorem}\label{th:efg-nz}
    For every $k \ge 1$, there exists a {\em nonzero-sum, tree-form, partially-observable} stochastic game with $\poly(k)$ states, and a pure Nash equilibrium, in which, for $\eps < 1/k$, there exist starting policies $\pi_1^0, \pi_2^0$ such that double oracle with adversarial meta-Nash equilibria takes $\Theta(2^k)$ iterations to converge.
\end{theorem}
\begin{proof}
    As before, we define the normal-form game first. In the {\em $n$-incrementing game}, two players simultaneously pick numbers $a_i \in [n]$. If $a_i = a_j+1$ then player $i$ scores $\alpha$ and player $j$ scores $-\beta$, where $\beta > \alpha > 0$. If $a_i = a_j$ then both players score $0$. Otherwise both players score a negative number.

    It is easy to see that, in the subgame where both players are restricted to $\{0, \dots, t\} \subseteq [n]$, $(t, t)$ is a Nash equilibrium (in fact, the unique welfare-maximizing equilibrium) and $t+1$ is a best response for both players. Thus, if both players are initialized at $\tilde\Pi_i^0 = \{0\}$, convergence will only happen after will only converge after $n$ iterations. We will set $n = 2^k$, and show that this game is representable as a stochastic game with $\poly(k)$ states.

    Consider the stochastic game defined as follows. Both players have action sets of size $2k$, identified with bitstrings consisting of completely repeated digits, \ie, 0, 1, 00, 11, 000, 111, \etc For cleanliness we will write $0^\ell$ to be the string with $0$ repeated $\ell$ times, and $1^\ell$ for the string with $1$ repeated $\ell$ times. These strings will denote the {\em trailing runs} of the players' bit strings. The transitions are as follows. At the root state, if both players play the same bit and different lengths, then both players score $-2$. If the players play different-length strings and neither player has played a string of length $1$, both players score $-2$. Otherwise, the game continues.

    At this point, there are three possibilities. From here onwards, players are forced to play either $0$ or $1$: any other action immediately terminates the game with both players scoring $-2$ (and is therefore dominated).
    \begin{enumerate}
        \item Both players have played $0^\ell$ or $1^\ell$. In this case, bit $i \in \{1, \dots, k-\ell-1\}$ is drawn uniformly at random and disclosed to both players, and both players choose an action. Both players score $0$ if the bits match, and $-1$ otherwise.
        \item One player has played $0^\ell$, and the other has player $1^\ell$. In this case, a bit $i \in \{ 1, \dots, k-\ell-1\}$ is drawn uniformly at random and disclosed to both players, and both players then choose an action again. The player who played $1^\ell$ scores $-1$. The player who  played $0^\ell$ scores $1/2k$ if the bits match, and $-1$ otherwise.
        \item One player (WLOG, P1) has played $0^\ell$ (for $\ell > 1$), and the other has played $1$. In this case, a bit $i \in \{1, \dots, k-2\}$ is selected at random. Then both players select an action. If $i =k-\ell$ then P1 is forced to play $1$; if $i>k-\ell$ then P1 is forced to play $0$. P1 scores $-1$. P2 scores $1/2k$ if the bits match, and $-1$ otherwise.
        \item One player (WLOG, P1) has played $1^\ell$ (for $\ell > 1$), and the other has played $0$. In this case, a bit $i \in \{ 1, \dots, k-2\}$ is selected at random. Then both players select an action. If $i = k-\ell$ then P1 is forced to play $0$. If $i > k-\ell$ then P1 is forced to play $1$. P2 scores $-1$. P1 scores $1/2k$ if the bits match, and $-1$ otherwise. 
    \end{enumerate}
A sketch of the game is depicted in \Cref{fig:efg-nz}.
Like the previous three proofs, we still essentially want a state machine to discriminate between the same five classes ($a_1-a_2 > 1, = 1, = 0, = -1, < -1$) but now we need the game to be tree-form. Since the comparison between the numbers requires knowing how long the trailing run of ones (or zeros) is, we ask the players for this information up-front---that is, both players at the start state choose the trailing runs of their numbers from the set $\{0, 1, 00, 11, 000, 111, \dots \}$ of size $2k$. Conditioned on these choices, the reward function is linear in the prefixes of the two players' bitstrings, and hence it can be represented by a single layer of the game tree in which a bit is selected at random and then the players pick assignments to that bit. 

An undominated pure policy (for either player) consists of a trailing run $0^\ell$ or $1^\ell$, and assignments to each bit $i \in \{1, \dots, k-\ell-1\}$. Thus, such strategies correspond exactly to the bitstrings in $\{0, 1\}^n$. It is easy to check that the utilities in the game restricted to undominated strategies satisfy the conditions of the $n$-incrementing game, completing the proof.
\end{proof}

Our final result will involve a case where both the meta-Nash equilibria {\em and} the best responses are not unique, and therefore we will assume that both are adversarially chosen. However, the game in the counterexample will have the most stringent structure: the counterexample is a {\em zero-sum tree-form, fully-observable stochastic game}.
\begin{theorem}\label{th:efg}
For every $k \ge 1$, there exists a zero-sum fully-observable, {\em tree-form} stochastic game with  $O(k)$ states and a Nash equilibrium of support size $2$ for each player in which, for $\eps < 2/k$, there exist starting policies $\pi_1^0, \pi_2^0$ such that double oracle with adversarial meta-Nash equilibria and best responses takes at least $2^{k-1}$ iterations. 
\end{theorem}
\begin{proof}
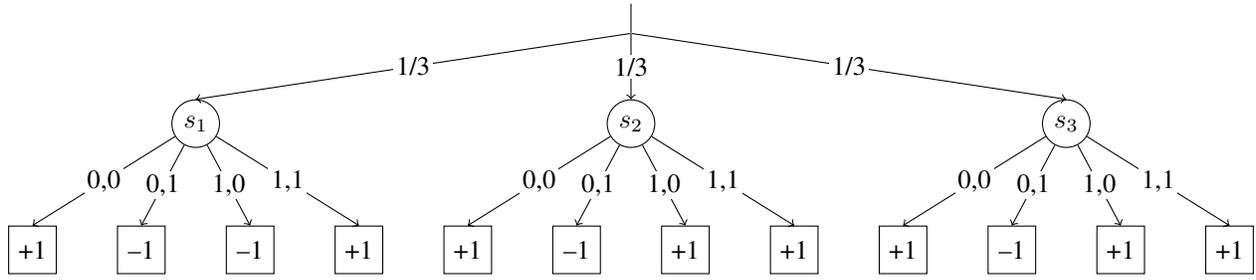
\begin{figure*}
\centering
\forestset{
  el/.style={edge label={node[midway, label] {#1}}},
  default preamble={for tree={ 
    l=1.75cm, child anchor=north, s sep=.8cm
  }},
}
\begin{forest}
[,draw=none,minimum size=0pt,for children={l=1em}
[,draw=none,minimum size=0pt,for children={l=1.25cm},edge={-},
  [$s_1$,el={1/3}
    [+1,terminal,el={0,0}]
    [--1,terminal,el={0,1}]
    [--1,terminal,el={1,0}]
    [+1,terminal,el={1,1}]
  ]
  [$s_2$,el={1/3}
    [+1,terminal,el={0,0}]
    [--1,terminal,el={0,1}]
    [+1,terminal,el={1,0}]
    [+1,terminal,el={1,1}]
  ]
  [$s_3$,el={1/3}
    [+1,terminal,el={0,0}]
    [--1,terminal,el={0,1}]
    [+1,terminal,el={1,0}]
    [+1,terminal,el={1,1}]
  ]
]
]
\end{forest}
\caption{A depiction of the game used in \Cref{th:efg}, for $k = 3$. Edges to the start states are labeled with their starting probabilities (1/3).}\label{fig:efg}
\end{figure*}
Unlike in the previous two proofs, in this proof it will be most convenient to start by defining the stochastic game without first discussing its normal form. Consider the following game. There are $k$ nonterminal states, $s_1, \dots, s_k$. The starting distribution $S_0$ is uniform on $\{s_1, \dots, s_k\}$. At each state, the players will each play a single action $a_i \in \{0, 1\}$, and then the game will end. It remains only to define the rewards. 
\begin{itemize}
\item At state $s_1$, P2 wins if and only if the players did not play the same action. That is, $s_1$ is a matching pennies game.
\item At state $s_j$ for $j > 1$, P2 wins if and only if P1 played $0$ and P2 played $1$. 
\end{itemize}
The winner gets value $+1$, and the loser gets value $-1$.

The equilibrium value of this game is $1-1/k$ for P1: the profile ``play uniform random at $s_1$ and $1$ at all other states'' is an equilibrium policy for both players of support size $2$. As before, we will identify pure strategies $\pi_i \in \{0, 1\}^{k}$ with the numbers they encode in binary. In this notation, let $\pi_1^0 = 2^{k}-1$ and $\pi_2^0 = 0$. Then we will show that, for $t \in \{1, \dots, 2^{k-1}-1\}$, the following adversarial choices of meta-Nash and best responses are possible in the double oracle algorithm:
\begin{enumerate}
\item $t-1$ is a best response for P1 against P2 playing $t-1$,
\item $t$ is a best response for P2 against P1 playing $2^{k}-1$,
\item $\tilde \Pi^t_1 = \{2^{k}-1 \} \cup \{ 0, \dots, t-1\}$, and $\tilde \Pi^t_2 = \{0, \dots, t\}$, and
\item $(2^{k}-1, t-1)$ is a meta-Nash equilibrium if $(\tilde \Pi^t_1, \tilde \Pi^t_2)$, that has equilibrium gap $2/k$ in the full game,
\end{enumerate}
We now prove all four points above by induction.
\begin{enumerate}
\item For P1, playing $t-1$ against $t-1$ wins all states, so it is a best response. 
\item Against $2^{k}-1$, P2 can only win the matching pennies game, which P2 does by playing any policy in the range $[0, 2^{k-1}-1]$. $t$ is indeed such a policy.
\item This follows from the previous two points and the definition of the double oracle algorithm.
\item The profile $(2^{k}-1, t-1)$ scores $1 - 2/k$ for P1 since P1 loses the matching pennies game but wins all others by playing $1$. P2 cannot improve upon this. P1 can only improve by winning at all states, but in order to do that, P1 must play a policu in the range $[t-1, 2^{k-1}-1]$. However, P1's policy set $\tilde \Pi^{t-1}_1$ only contains $\{0, \dots, t-2\}$ by induction hypothesis, so P1 cannot win all states, and therefore $(2^{k}-1, t-1)$ is a meta-Nash equilibrium.
\end{enumerate}
This completes the induction and therefore the proof, since with these choices, the Nash gap computed by double oracle will stay at $2/k$ until at least iteration $2^{k-1}$. 
\end{proof}

\section{Discussion and Related Work}

In this section, we discuss a few alternative algorithms similar to the double oracle algorithm, and how they relate to the results in this paper.

\subsection{Fictitious Play}
Another common algorithm for reducing multi-player to single-player games is {\em fictitious play}. Fictitious play differs from double oracle only in the choice of opponent policies $\mu_{-i}^t$ against which player $i$ computes the best response $\pi_i^t$. While double oracle uses a Nash equilibrium of the restricted game defined by the policies already discovered, fictitious play uses a simple uniform average over those policies:
\begin{align}
    \mu_{-i}^t := \frac{1}{t}\sum_{\tau = 0}^{t-1} \pi_{-i}^{(\tau)}. 
\end{align}
Although this change seems simple, the two algorithms behave very differently in theory. For example, double oracle is guaranteed to converge in at most $|\Pi|$ iterations, where $\Pi$ is the set of policies, since at least one policy is added on every iteration until convergence is reached. However, proving (or disproving) a $\poly(|\Pi|, 1/\eps)$-time convergence rate for fictitious play, even in zero-sum games is one of the oldest open problems in game theory, known as {\em Karlin's conjecture}~\cite{Karlin59:Mathematical}. Similarly to our discoveries, however, the behavior of fictitious play is known to depend on assumptions about tiebreaking. In particular, it is known that for normal-form games whose payoff matrix is diagonal, the convergence rate of fictitious play is polynomial if the best responses are chosen using a consistent tiebreaking method~\cite{Abernethy21:Fast}, but not if they are chosen adversarially~\cite{Daskalakis14:Counter}.

\subsection{$\alpha$-Best Response Dynamics and Potential Games}\label{sec:discussion}

In {\em best response dynamics}, we simply set $\mu_{-i}^t = \pi_{-i}^{t-1}$. That is, each player simply best responds to the opponent's previous policy. In zero-sum games, best response dynamics usually will not converge to equilibria: indeed, since $\pi^t$ is always pure, best response dynamics cannot converge whenever there is no pure equilibrium. However, best response dynamics have been considered in the class of {\em potential games}, which are, roughly speaking, games that ``look like'' ones in which every player has the same utility function. In this class of games, it has been observed~\cite{Awerbuch08:Fast,Chien11:Convergence} that it is sometimes better to {\em limit} players to only playing best responses if they improve the player's utility by more than some parameter $\alpha$. 

One may ask whether a similar change affects our lower bounds. That is, suppose that, in the double oracle algorithm, the best response $\pi_i^t$ is only added to $\Pi_i^t$ if $V(\pi_i^t, \mu_{-i}^t) - V(\mu^t) \ge \alpha$, where $\eps \ge \alpha > 0$. Let us call this algorithm {\em $\alpha$-double oracle}.

\begin{itemize}
    \item In \Cref{th:gmp}, the best response of P1 at iteration $2t$ improves the value by $2/t$, and the best response of P2 at iteration $2t+1$ improves the value by a full $2$. Thus, the theorem is unaffected.
    \item In \Cref{th:posg}, \Cref{th:sg}, and \Cref{th:efg-nz}, the value improvement of every player on every iteration is equal to the Nash gap. Therefore, these results are unaffected.
    \item \Cref{th:efg} {\em is} affected. That result relies on the ability for P2 to add the best response $\pi_2^t = t$, which does not improve P2's value at all. Thus, the result breaks for every $\alpha > 0$.
\end{itemize}

\section{Conclusions and Future Research}
We have shown, to our knowledge, the first exponential lower bounds on the convergence time (in number of iterations) of the double oracle algorithm. We leave several natural questions for future research.
\begin{itemize}
    \item Can the gaps in \Cref{tab:summary} be closed? For example, does there exist a tree-form POSG in which the double oracle algorithm must take exponentially many iterations with any of the adversarial assumptions removed? Does there exist a fully-observable stochastic game in which the double oracle algorithm is exponential even with non-adversarial best responses?
    \item Are there ``simple'' modifications to double oracle, for example, $\alpha$-double oracle as suggested in \Cref{sec:discussion}, that guarantee polynomial worst-case bounds in certain cases (\eg, zero-sum tree-form games)?
\end{itemize}

\section*{Acknowledgements}

This work is supported by the Vannevar Bush Faculty Fellowship ONR N00014-23-1-2876, National Science Foundation grants RI-2312342 and RI-1901403, ARO award W911NF2210266, and NIH award A240108S001. Brian Hu Zhang's work is supported in part by the CMU Computer Science Department Hans Berliner PhD Student Fellowship.

\newpage
\bibliographystyle{named}
\bibliography{dairefs}
\end{document}